   \documentclass[10pt]{article}

\usepackage{amsmath}
\usepackage{amssymb}
\usepackage{latexsym} 
\usepackage[all]{xypic}
\usepackage[all]{xy}
\usepackage{listings}
\usepackage{amsfonts}
\usepackage[colorlinks=true,citecolor=red]{hyperref}
\usepackage{mathrsfs}
\usepackage{mytheorems}
\newcommand{\cutout}[1]{}
\newcommand{\card}[1]{|#1|}
\newcommand{\Ent}[1]{{\cal E}(#1)}
\newcommand{\ET}[1]{\widetilde{{\cal E}}(#1)}

\newcommand{\minus}[1]{{#1}^{-}}

\newcommand{\rTo}{\longrightarrow}

\newcommand{\set}[1]{\{\,#1\,\}}

\newcommand{\tif}{\text{ if }}
\newcommand{\tand}{\text{ and }}

\newlength{\myhidel}

\newcounter{casecounter}
\def\thecasecounter{(\roman{casecounter})}
\newenvironment{proofbycases}[1][0]{
  \setcounter{casecounter}{#1}%
}{}

\newenvironment{caseinproof}{%
  \refstepcounter{casecounter}%
  \par\noindent%
  \emph{Case~\thecasecounter.}~
}{}

\newcommand{\firstvariable}[1]{\ifthenelse{\equal{#1}{a}}{a_{0}}{b_{0}}}
\newcommand{\secondvariable}[1]{\ifthenelse{\equal{#1}{a}}{a_{1}}{b_{1}}}
\newcommand{\thirdvariable}[1]{\ifthenelse{\equal{#1}{a}}{a_{2}}{b_{2}}}

\newcommand{\secondplayer}[1]{
  \ifthenelse{\equal{#1}{\pi}}{\sigma}{
    \ifthenelse{\equal{#1}{\sigma}}{\pi}{
      \ifthenelse{\equal{#1}{M}}{O}{M}
      }
    }
}

\ifx\TAC\undefined
\RequirePackage{ifthen}
\newcommand{\Whitstar}{}
\newcommand{\Whitvariablea}{a_{0}}
\newcommand{\Whitvariableb}{a_{1}}
\newcommand{\Whitvariablec}{a_{2}}
\newcommand{\Whitlabel}[3]{
  \ifthenelse{\equal{#2}{\pi}}{\renewcommand{\Whitstar}{}}
  {\renewcommand{\Whitstar}{*}}
  \ifthenelse{\equal{#3}{a}}{
   \renewcommand{\Whitvariablea}{a_{0}}
   \renewcommand{\Whitvariableb}{a_{1}}
   \renewcommand{\Whitvariablec}{a_{2}}
   }
  {\renewcommand{\Whitvariablea}{b_{0}}
   \renewcommand{\Whitvariableb}{b_{1}}
   \renewcommand{\Whitvariablec}{b_{2}}}
 W_{#1}^{\Whitstar}(\Whitvariablea,\Whitvariableb,\Whitvariablec)
}

\newcommand{\classe}{\mathcal{C}}

\title{Closure Under Minors  of Undirected Entanglement}

\author{Walid Belkhir \\
\begin{small} Laboratoire d'Informatique Fondamentale de Marseille \end{small} \\ 
 CMI, 39 rue F. Joliot Curie 13453 Marseille Cedex 13, France  \\
 \texttt{belkhir@cmi.univ-mrs.fr}
} 

\date{}
\begin{document}
\maketitle

\begin{abstract}
Entanglement is a digraph complexity measure  that origins in fixed-point theory. Its  purpose is to  count the nested depth of cycles in digraphs. 
  In this paper we   prove  that the class of undirected graphs of entanglement at most  $k$, for arbitrary fixed $k \in \mathbb{N}$, is closed under taking  minors.  Our proof relies on the game theoretic characterization of entanglement in terms of Robber and Cops games. 
\end{abstract}
\section{Introduction}
Entanglement is a complexity measure of finite directed graphs introduced in \cite{berwanger} as a tool to analyze the descriptive complexity of the propositional modal $\mu$-calculus \cite{BerwangerGraLen06}.  This measure has shown its use in solving the variable hierarchy problem\footnote{This problem asks whether the expressive power of a given fixed-point logic increases with the number of bound variables.} for the modal $\mu$-calculus \cite{BerwangerGraLen06} and  for the lattice $\mu$-calculus \cite{LPAR08}.  Roughly speaking, the entanglement of a $\mu$-formula (viewed as a graph)  gives the minimum number of bound variables (i.e. fixed-point variables) required in any equivalent $\mu$-formula. From these considerations, the entanglement is considered as the combinatorial counterpart of  the variable hierarchy.

Leaving the logical motivations in the background, recent works have been devoted to a  graph theoretic study of entanglement \cite{BelkSanto0a7,Rabthese,CAI09}, and in particular to characterizing the structure of graphs of  entanglement at most $k$. However,  only partial results are known:  the structure of directed graphs for $k=1$  \cite{berwanger}, $k=2$ \cite{Rabthese}; and of  undirected graphs for $k=2$  \cite{BelkSanto0a7},   and partially for $k=3$ \cite{CAI09}. \\
Furthermore, the exact complexity of deciding the entanglement of a graph is not yet known. By using general algorithms \cite{Jur00:STACS} it was argued in \cite{mathese} that deciding whether a graph has entanglement at most $k$, for fixed $k$, is a problem in PTIME.  In particular,  using the structural characterizations  mentioned above,  this   problem is in NLOGSPACE  for  directed graphs and $k=1$ \cite{berwanger}.    This problem can be solved in linear time  for the undirected graphs and $k=2$  \cite{BelkSanto0a7} and in cubic time  for directed graphs and $k=2$ \cite{Rabthese}.

In this paper we prove a fundamental result of the undirected entanglement: the class of undirected graphs of bounded entanglement is closed under minors. Our working definition of the entanglement of a graph $G$  is the minimum number of $k$-cops  required to   catch  Robber in some games $\Ent{G,k}$ on $G$ \cite{berwanger}.     Since the other definition \cite{berwanger} in terms of a certain unfolding  into trees with back edges can not be used  in an easy way.   Our proof technique to show that the  entanglement of a (undirected) graph $H$ is greater than the entanglement of its minor $G$  is largely inspired by  \cite{Joyal97,FraisseBook}:  a move of Robber in the game $\Ent{G,k}$ is simulated by a move or a sequence of Robber's moves in the game $\Ent{H,k}$ and in turn, Cops' response in $\Ent{H,k}$ is mapped to $\Ent{G,k}$ in the desired way   and soon. This sort of back-and-forth simulation reminds the back-and-forth games of \cite{FraisseBook}.

 Wagner's conjecture (proved in a series of papers by Robertson and Seymour \cite{GraphMinors10}), states that for every infinite set of graphs, one of its members is a minor of an other. Thus, every class of graphs that is closed under taking minors can be characterized by a finite set of excluded minors.    Since the class of undirected graphs of  bounded entanglement  is minor closed,  Theorem \ref{minorclosureTh},  then  it follows   that this class can be characterized   by a finite set of excluded minors. Therefore,   testing weather an  undirected graph has entanglement at  most $k$ can be checked in cubic time.     

 Finally we point out that only the set of excluded minors characterizing the   graphs of entanglement $\le 2$ is known  \cite{BelkSanto0a7},  or see \cite[\S 7]{mathese} for more details.  In the case of entanglement,  the  number of excluded  minors is relatively large   because an excluded minor  may contain articulation points.  The main challenge consists in finding a compact representation of the excluded minors.

\paragraph{Preliminaries and notations} ~\\
Throughout this paper, an undirected graph is called simply a graph, and a directed graph is  called a digraph.  A graph $G$ is a \emph{minor}  of a graph $H$  if $G$ can be obtained from $H$ by successive application of the following operations on it: \emph{(i)} delete an edge, \emph{(ii)} contract an edge, \emph{(iii)} delete an isolated vertex.

  Given a graph $G$ and an edge $e$, \emph{edge deletion} results in a graph $G\setminus e$ with the same vertex set as $G$ and the edge set $E_G \setminus \{e\}$; \emph{edge contraction} results in a graph $\partial_{e}^{z}G$ with the vertex set obtained by replacing the end-vertices of $e$ in $G$ by a new vertex $z$, the latter inherits all the neighbors of the two replaced vertices.  We shall write $\mathcal{N}(v)$ for the neighbors of vertex $v$. We denote by $G\setminus v$ the vertex deletion.

A class $\classe$ of graphs is closed under minors if  $G \in \classe$ then for every minor $H$ of $G$ we have that $H \in \classe$.      
% A graph $G$ is said to be undirected if whenever $(v,w) \in E_G$ then $(w,v) \in E_G$.

\section{Entanglement}\label{sec:entang}
 The entanglement of a finite digraph $G$, denoted $\Ent{G}$, was
defined in \cite{berwanger} by means of some games $\Ent{G,k}$, $k =
0,\ldots ,\card{V_{G}}$. The game $\Ent{G,k}$ is played on the graph
$G$ by Robber against Cops, a team of $k$ cops. The rules are as
follows. Initially all the cops are placed outside the graph, Robber
selects and occupies an initial vertex of $G$.  After Robber's move,
Cops may do nothing, may place a cop from outside the graph onto the
vertex currently occupied by Robber, may move a cop already on the
graph to the current vertex.  In turn Robber must choose an edge
outgoing from the current vertex whose target is not already occupied
by some cop and move there.  If no such edge exists, then Robber is
caught and Cops win.  Robber wins if he is never caught.  It will be useful to formalize
these notions. \\

\begin{definition}
\emph{
  The entanglement game $\Ent{G,k}$ of a digraph $G$ is defined by:
\begin{itemize}
  \item Its positions are of the form $(v,C,P)$, where $v \in V_{G}$,
    $C \subseteq V_{G}$ and $\card{C} \leq k$, $\hspace{2mm} P \in \{Cops,
    Robber\}$.
 \item  Initially Robber chooses $v_{0} \in V_G$ and moves to
    $(v_0,\emptyset,Cops)$. 
 \item Cops can move  from $(v,C,Cops)$ to $(v,C',Robber)$
    where $C'$ can be
     \begin{enumerate}%[$\bullet$]
     \item   $C$ : Cops skip,
       \item $C \cup\set{v}$ : Cops add a new Cop on the
      current position,
     \item  $(C \setminus\set{x}) \cup \set{v}$ : Cops move a placed Cop
      to the current position.
\end{enumerate}
 \item Robber can move from $(v,C,Robber)$ to $(v',C,Cops)$ if
    $(v,v') \in E_{G}$ and $v' \notin C$.
\end{itemize}
  Every finite play is a win for Cops, and every infinite play is a win
  for Robber. 
}
\end{definition}

\emph{The entanglement of $G$, denoted by $\Ent{G}$,  is the minimum $k \in \set{
    0,\ldots ,\card{V_{G}}}$ such that Cops have a winning strategy in
  $\Ent{G,k}$}.  \\

The following Proposition provides a useful variant of entanglement games, see also \cite{LPAR08}.  
\begin{proposition} \label{modif:entag}
 Let $\ET{G,k}$ be the game played as the game $\Ent{G,k}$ apart
  that Cops are allowed to retire a number of cops placed on the
  graph. That is, Cops moves are of the form
\begin{itemize}
  \item  $(g,C,Cops) \rightarrow (g,C',Robber)$     (generalized skip move),
  \item  $(g,C,Cops) \rightarrow (g,C'\cup \set{g},Robber)$     (generalized replace move), 
\end{itemize} 
 where in both cases $C' \subseteq C$.
  Then Cops have  a winning strategy in $\Ent{G,k}$ if and only if they 
  have a winning strategy in $\ET{G,k}$.
\end{proposition}
\cutout{
\begin{proof}
  Since every Cops' move in the game $\Ent{G,k}$ is a Cops' move in
  the game $\ET{G,k}$, and since there is no new kind of moves for
  Robber in the game $\ET{G,k}$, then a Cops' winning strategy in
  $\Ent{G,k}$ can be used to let Cops win in $\ET{G,k}$.

  In the other direction, a winning strategy for Cops in $\ET{G,k}$
  gives rise to a winning strategy for Cops in $\Ent{G,k}$ as
  follows.

  Each position $(g,C,P)$ of $\Ent{G,k}$ is matched by a position
  $(g,\minus{C},P)$ of $\ET{G,k}$ such that $\minus{C} \subseteq C$. A
  Robber's move $(g,C,Robber) \rightarrow (g',C,Cops)$ in $\Ent{G,k}$
  shall be simulated by the move $(g,\minus{C},Robber) \rightarrow
  (g',\minus{C},Cops)$ in $\ET{G,k}$. Note that Robber can perform such
  a move, since if $g' \in \minus{C}$ then already $g' \in C$.

  Assume that the position $(g,C_{0},Cops)$ of $\Ent{G,k}$ is matched by
  the position $(g,\minus{C}_{0},Cops)$ of $\ET{G,k}$. From
  $(g,\minus{C}_{0},Cops)$, Cops' winning strategy may suggest two kinds
  of moves.
  It may suggest a generalized skip $(g,\minus{C}_{0},Cops)
  \rightarrow (g,\minus{C}_{1},Robber)$ with $\minus{C}_{1} \subseteq
  \minus{C}_{0}$.  In this case, Cops skip on from the related
  position $(g,C_0,Cops)$, so that the new position $(g,C_0,Robber)$ is
  matched by $(g,\minus{C}_{1},Robber)$.

  Otherwise, Cops' winning strategy in $\ET{G,k}$ may suggest a
  generalized replace move $(g,\minus{C}_{0},Cops) \rightarrow
  (g,\minus{C}_{1} \cup \set{g},Robber)$.  If $\card{C_0} < k$, then
  Cops perform the add move $(g,C_0,Cops) \rightarrow (g,C_0 \cup
  \set{g},Robber)$. Notice that $\minus{C}_{1} \cup \set{g} \subseteq
  \minus{C}_{0} \cup \set{g} \subseteq C_{0} \cup \set{g}$.
  If $\card{C_0} = k$, then observe that $C_0 \setminus \minus{C}_ 1$ is not empty:
  we have $\minus{C}_1
  \subseteq \minus{C}_{0} \subseteq C_0$
  and
  $\card{\minus{C}_{1}} < k$, since $g \not \in \minus{C}_{1}$ and
  $\card{\minus{C}_{1} \cup \set{g}} \leq k$.
%   hence  and $\card{C^-_1} <
%   \card{C_0}$.
  Consequently we can pick $x \in C_0 \setminus \minus{C}_{1}$ such
  that $x \neq g$, since $g \not\in C_{0}$. Therefore Cops simulate
  the move $(g,\minus{C}_{0},Cops) \rightarrow (g,\minus{C}_{1} \cup
  \set{g},Robber)$ of $\ET{G,k}$ with the replace move $(g,C_0,Cops)
  \rightarrow (g,C_0\setminus \set{x} \cup \set{g},Robber)$ on
  $\Ent{G,k}$. Notice again that the invariant $\minus{C}_{1} \cup
  \set{g} \subseteq C_0 \setminus \set{x} \cup \set{g}$ is maintained.
\end{proof}
}

\section{Closure under minor of undirected entanglement}

 \begin{lemma} \label{subgraph:lemma}
If $G$ is a subgraph of $H$ then $\Ent{G} \le \Ent{H}$.
\end{lemma}
\begin{proof}
Let $k=\Ent{G}$, then clearly, if Robber has a winning strategy in $\Ent{G,k}$ then he can use it to win in $\Ent{H,k}$ by restricting his moves on $G$.
\end{proof}

\begin{theorem}\label{minorclosureTh}
  The class of graphs of entanglement at most $k$, for arbitrary fixed $k \in \mathbb{N}$, is minor closed, that is if $G$ is  a  minor of $H$ then $\Ent{G}\le \Ent{H}$. 
\end{theorem}
\begin{proof}
  If $G$ is obtained from $H$ by edge-deletion then the statement obviously holds  by  Lemma \ref{subgraph:lemma}. Otherwise, if $G$ is obtained by edge-contraction i.e.  $G=\partial_{ab}^{z}H $ for some $ab \in E_H$, then  this allows to define a total function $f: V_H \rTo V_{G}$  as follows: 
\begin{align*}
f(v)=  \left \{
\begin{array}{ll}
  z & \tif v\in \set{a,b}, \\
  v  & \textrm{ otherwise}.
\end{array}
\right.
\end{align*}
 Let $k=\Ent{H}$, using the function  $f$ we shall construct a Cops' winning strategy in the game $\ET{G,k}$  out of a  Cops' winning strategy in $\Ent{H,k}$.
\noindent
To this goal, every position  $(g,C_G,P)$ of  $\ET{G,k}$ is matched with the position   $(h,C_H,P)$ of $\Ent{H,k}$,  where $P \in \set{Robber,Cops}$, such that the following invariants  hold:
\begin{align*}
  & \bullet \textrm{} g=f(h) \tand C_G=f(C_H), \label{COPS} \tag{COPS}  \\
  & \bullet  if  \textrm{ }  g=z   \textrm{ (hence }  h \in \set{a,b}) \tand P=Robber,    \textrm{ then } \\ & \hspace{7mm} z \in C_G \tand  h \in C_H;  \textrm{ moreover } \card{C_H \cap \set{a,b}}=1  \label{THIEF-Z} \tag{Robber-Z}. 
\end{align*}
\noindent The invariant (Robber-Z) may be understood as follows: whenever Robber will move from $z$ then $z$ must be occupied by a cop. At this moment, in $\Ent{H,k}$, either $a$ or $b$ must be occupied by a cop but not both.

We simulate every Robber's move of the form 
\begin{align*} M_G=(v,C_G,Robber) \to (w,C_G,Cops) \end{align*} of $\ET{G,k}$ either by a move or a sequence of moves in $\Ent{H,k}$ according to the \emph{locality} of Robber's move $M_G$:
\begin{enumerate}
\item  If  $M_G$ is \emph{outside $z$}, i.e. $v,w \neq z$ then in this case $M_G$ is simulated by the same move in $\Ent{H,k}$. 
\item  If $M_G$ is \emph{entering} to $z$, i.e. $w=z$ and $vw \in E_G$.  Assume $v \in \mathcal{N}(a)$ \footnote{The case $v \in \mathcal{N}(b) \setminus \mathcal{N}(a)$ is similar;  recall that $\mathcal{N}(v)$ are just the neighbors  of $v$.}. In this case, the move $M_G$ is  simulated  by a finite alternation   of Robber  between $a$ and $b$ until Cops put a cop on $a$ or $b$, and then the simulation is halted. That is,  the move $M_G$ is  simulated  by the finite alternating  sequence $M^{\star}_H$ of moves that is the following sequence apart the last move:

\begin{align*}
M^{\star}_{H}=(v,C_H,Robber)\to (a,C_H,Cops) & \to (a,C_H,Robber) \to (b,C_H,Cops)\\
 & \to (b,C_H,Robber)\to (a,C_H,Cops) \\
 & \to  \dots  \\
 &\to (x,C_H,Robber) \to (y,C_H,Cops)  \\ 
 & \hspace{9mm} "M^{\star}_H \textrm{ ends here}"\\
 & \to (y,C'_H,Robber)  \;\;\; 
\end{align*}

Such that $\set{x,y}=\set{a,b}$ and $C'_H \neq C_H$. Clearly $y \in C'_H$. Observe that this sequence is possible i.e. $b\notin C_H$, because if $b \in C_H$ then it follows by the invariant (\ref{COPS}) that $f(b)=z \in f(C_H)=C_G$,  that is $z \in C_G$,  which can not happen because we have assumed that the move $M_G$ is possible. The particular  case of Robber's first move to $z$ is simulated  by a similar finite alternating sequence of moves between $a$ and $b$, apart that  $C_H=C_G=\emptyset$.  
\item If $M_G$ is \emph{leaving} $z$, i.e. $v=z$ and $vw \in E_G$. Assume that the position $(z,C_G,Robber)$ is matched with $(a,C_H,Robber)$. Recall that  $z \in C_G$ and $a\in C_H$, by the invariant (Robber-Z). \\
\begin{enumerate}%[{3.}1]
\item  If $w \in \mathcal{N}(a)$ then the move $M_G$ is simulated by the same move of $\Ent{H,k}$.
 \item  If $w \in \mathcal{N}(b) \setminus \mathcal{N}(a)$, then the move $M_G$ is simulated by the  following sequence of moves:
 \begin{align*}
(a,C_H,Robber) \to (b,C_H,Cops) \to (b,C'_H,Robber) \to (w,C'_H,Cops).
 \end{align*}
 This sequence is possible, i.e. $b \notin C_H$ because already $a\in C_H$, therefore $b\notin C_H$, by the invariant (Robber-Z).
At this point, the  ending position of  $M_G$ -- which is the position $(w,C_G,Cops)$  -- is matched  with  the position $(w,C'_H,Cops)$ of $\Ent{H,k}$, we emphasize that Cops' next move  $(w,C'_H,Cops) \to (w,C''_H,Robber)$  in $\Ent{H,k}$  should be mapped to the move 
\begin{align*} 
(w,C_G,Cops) \to (w,f(C''_H),Robber) 
\end{align*}
in $\ET{G,k}$, and the main technical part  is to prove that the  latter move  respects the rules of the game .   
\end{enumerate}
\end{enumerate}
 A Cops' move  in $\Ent{H,k}$ is mapped to a Cops' move  in $\ET{G,k}$ as follows. 
Assume that the position  $(g,C_G,Cops)$ of $\ET{G,k}$ is matched with the position $(h,C_H,Cops)$ of $\Ent{H,k}$ and  moreover Cops have moved to 
\begin{align}\label{MV1}
(h,C_H,Cops)  \to (h,C'_H,Robber)
\end{align}
Therefore  Cops in $\ET{G,k}$ should move to 
\begin{align}\label{MV2}
(g,C_G,Cops) \to (g,f(C'_H),Robber)
\end{align}
 the aim is prove that the latter  move is legal w.r.t the rules of the game $\ET{G,k}$.  We distinguish three cases according to the manner by which $g$ has been reached by Robber in $\ET{G,k}$ in the previous round of simulation.
\begin{enumerate}
\item If $g$ has been reached by  an \emph{outside move}, hence $g \neq z, g=h$ ($g$ is the vertex considered in the move  (\ref{MV2}), and $h$ is considered in the  move  (\ref{MV1}), then  in this case, $C'_H$ may be written: $C'_H=(C_H\setminus A) \cup B$, where $\emptyset \subseteq B \subseteq \set{g}$  and $\card{A} \le 1$. (To be more precise we have $\card{A}\le \card{B}$.) Therefore
\begin{align*}
f(C'_H)&=[f(C_H \setminus A)] \cup f(B) \\
         &   = \left\{ \begin{array}{ll}
                & f(C_H) \cup f(B) \hspace{3mm}  \tif a,b \in C_H \tand A \subseteq \set{a,b},\\
            & (f(C_H) \setminus f(A)) \cup f(B) \hspace{3mm}  \textrm{ otherwise}
   \end{array}
\right.
\end{align*}
It is easy to  see that this is a legal move.
\item If $g$ has been reached by  an \emph{entering move}, hence $g=z$ and $h \in \set{a,b}$ (again $g$ is the vertex considered in the move  (\ref{MV2}), and $h$ is considered in the move (\ref{MV1})), then in this case $z\notin C_G$ and therefore $a,b \notin C_H$.  We shall argue  that the move $(z,C_G,Cops)\to (z,f(C'_H),Robber)$ respects the rules of the game. Assume that $h=a$. In this case $C'_H$ is of the form  
\begin{align*}
C'_H=(C_H \setminus A) \cup B
\end{align*}
where $ 0 \le |A| \le 1$ with $a,b \notin A$ and $\emptyset \subseteq B \subseteq \set{a}$, therefore 
 \begin{align*}
f(C'_H)&=f[(C_H \setminus A) \cup B] \\
     & = [f(C_H) \setminus f(A)] \cup f(B)
\end{align*}
Observe that  $z \notin f(A)$ and $\emptyset \subseteq f(B) \subseteq \set{z}$. Hence the move in question  respects the rules of the game.
\item If $g$ has been reached by a  \emph{leaving move}, hence  $h=g$ and $hz \in E_G$, then  in this case $z \in C_G$ and either $a \in C_H$ or $b \in C_H$ but not both, by the invariant (\ref{THIEF-Z}). We distinguish two cases:
\begin{proofbycases}
\begin{caseinproof} If $h$ has been reached by a single Robber's move in $\Ent{H,k}$ in the previous round of simulation, then one can check easily that every Cops' move  from position $(h,C_H,Cops)$ in $\Ent{H,k}$ is mapped to the same move from $(h,C_G,Cops)$ in $\ET{G,k}$.
\end{caseinproof}
\begin{caseinproof} If $h$ has been reached by a sequence of  moves in $\Ent{H,k}$, then let us go back to the previous round of the simulation. The previous  move in $\ET{G,k}$ was indeed of the form 
\begin{align*}
(z,C_G,Robber) \to (h,C_G,Cops)
\end{align*}
and its related simulation moves in $\Ent{H,k}$ are of the form
\begin{align*}
(a,C^{-1}_H,Robber) \to (b,C^{-1}_H,Cops) \to (b,C_H,Robber) \to (h,C_H,Cops) 
\end{align*}
In $\Ent{H,k}$, if Cops move to $(h,C_H,Cops) \to (h,C'_H,Robber)$ then this move is obviously mapped to Cops' move $(h,C_G,Cops) \to (h,f(C'_H),Robber)$ in $\ET{G,k}$. 
 Note that  $C'_H=(C^{-1}_H\setminus A) \cup B$ where \\ $\emptyset \subseteq B \subseteq \set{b,h}$ and $A \subseteq V_H$ with $0\le \card{A} \le 2$, let us compute $C'_G=f(C'_H)$ in terms of $C_G$:
\begin{align*}
f(C'_H) &=[f(C^{-1}_H \setminus A)] \cup f(B) \\
           & = [\big(f(C^{-1}_H) \setminus f(A)\big) \cup Z] \cup f(B) 
\end{align*}

  where $\emptyset \subseteq Z \subseteq \set{z}$   and 
  $\emptyset \subseteq f(B)=B' \subseteq \set{z,h}$, therefore 
\begin{align*}
f(C'_H) & = [f(C^{-1}_H) \setminus f(A)] \cup (Z \cup B')\\
        &  =(C_G \setminus f(A)) \cup B''
\end{align*}

 where still $\emptyset \subseteq B''=Z \cup B' \subseteq \set{z,h}$. Recall that $z \in C_G$ by the invariant (\ref{THIEF-Z}) and hence the move in question respects the rules of the game. 
\end{caseinproof}
\end{proofbycases}
\end{enumerate}

 Finally, the invariants  (\ref{COPS}) and (\ref{THIEF-Z})  are preserved by construction.  This ends the proof of Theorem \ref{minorclosureTh}.
 \end{proof}
   A similar  Proposition to  the following  one  concerning the  tree-width instead of the entanglement, has been   proved in \cite{GraphMinors2}.
\begin{proposition}
If $G$ is a direct minor of $H$ then $\Ent{H}-1 \le \Ent{G}$
\end{proposition}
\begin{proof}
We need the following Claim.
\begin{claim} \emph{To prove that $\Ent{H}-1 \le \Ent{G}$ it suffices to prove that $\Ent{H \setminus v} \le \Ent{G}$, for some $v\in V_H$}.
\end{claim}
\begin{proof}
Assume that $\Ent{H\setminus v} \le \Ent{G}$, and let $k=\Ent{G}$. This implies that if Cops have a winning strategy  in $\Ent{G,k}$ then they have a winning strategy $S_1$   in $\Ent{H\setminus v,k}$. Out of the winning strategy $S_1$ they can construct a winning strategy in $\Ent{H,k+1}$ as follows: if Robber  restricts his moves on  $V_H \setminus v$ then play with $S_1$, and if Robber goes to $v$ then put the  $(k+1)^\textrm{th}$ cop on $v$ and never move it.  This ends the proof of the Claim.
\end{proof}
If   $G$ is obtained from $H$ by deleting some edge $ab$, then observe that $H \setminus a $ is a subgraph of $G$, therefore from Lemma \ref{subgraph:lemma} we get $\Ent{H \setminus a} \le \Ent{G}$.  We conclude -- according to the Claim -- that $\Ent{H}-1 \le \Ent{G}$. If   $G$ is obtained from $H$ by contracting some edge $ab$, then  $H\setminus a$ is  again a subgraph of $G$, and  the argument  is similar to the above one. 
\end{proof}

The following Corollary  provides a useful indication for  searching    the minimal set of excluded minors characterizing  graphs of bounded entanglement.
\begin{corollary}
Let $\mathcal{F}_k$ be the  minimal excluded minors for the  class of graphs of entanglement at most $k$. Then,  every graph in $\mathcal{F}_k$ has exactly entanglement $k+1$.
\end{corollary}

\bibliographystyle{plain}
\bibliography{biblio}

\end{document}

%%% Local Variables: 
%%% mode: latex
%%% TeX-master: 
%%% End: 